\documentclass[a4paper,UKenglish]{lipics-v2016}
\usepackage{tikz}
\usetikzlibrary{automata,positioning}

\theoremstyle{plain}
\newtheorem{proposition}[theorem]{Proposition}

\DeclareMathOperator{\lcm}{lcm}
\newcommand{\wt}{\widetilde}
\newcommand{\A}{\mathcal{A}}
\newcommand{\Z}{\mathbb{Z}}
\newcommand{\N}{\mathbb{N}}
\newcommand{\Q}{\mathbb{Q}}
\newcommand{\F}{\mathcal{F}}
\newcommand{\LL}{\mathcal{L}}
\newcommand{\SL}{\mathrm{SL}(2,\Z)}

\renewcommand{\phi}{\varphi}

\newcounter{thm1}
\newcounter{thm2}
\newcounter{thm3}
\newcounter{thm4}
\newcounter{thmold}

\title{Vector Reachability Problem in $\SL$ \footnote{This work was partially 
supported by EPSRC grant ``Reachability problems for words, matrices and maps'' (EP/M00077X/1).}}

\author[1]{Igor Potapov}
\author[2]{Pavel Semukhin}
\affil[1]{Department of Computer Science, University of Liverpool, United Kingdom\\
  \texttt{potapov@liverpool.ac.uk}}
\affil[2]{Department of Computer Science, University of Liverpool, United Kingdom\\
  \texttt{semukhin@liverpool.ac.uk}}
\authorrunning{I. Potapov and P. Semukhin} %mandatory. First: Use abbreviated first/middle names. Second (only in severe cases): Use first author plus 'et. al.'

\Copyright{Igor Potapov and Pavel Semukhin}%mandatory, please use full first names. LIPIcs license is "CC-BY";  http://creativecommons.org/licenses/by/3.0/

% mandatory: Please choose ACM 1998 classifications from http://www.acm.org/about/class/ccs98-html . E.g., cite as "F.1.1 Models of Computation". 
% mandatory: Please provide 1-5 keywords
\ArticleNo{\null}
\begin{document}

\maketitle

\begin{abstract}
The decision problems on matrices were intensively studied for many decades as
matrix products play an essential role in the representation of various
computational processes.
However, many computational problems
for matrix semigroups are inherently difficult to solve even for problems
in low dimensions and most matrix semigroup problems become undecidable
in general starting from dimension three or four.

This paper solves two open problems about the decidability
of the vector reachability problem over a finitely generated semigroup of matrices
from $\SL$ and the point to point reachability (over rational numbers)
for fractional linear transformations, where associated matrices
are from $\SL$. The approach to solving reachability problems is
based on the characterization of reachability paths between points which is followed
by the translation of numerical problems on matrices into computational
and combinatorial problems on words and formal languages.
We also give a geometric interpretation of reachability paths and extend the
decidability results to matrix products represented by arbitrary
labelled directed graphs.
Finally, we will use this technique to prove that a special case of the scalar reachability problem is decidable.
 \end{abstract}

\section{Introduction}
Decision problems on matrices were intensively studied from 1947 
when A. Markov showed the connection between classical computations and
problems for matrix semigroups \cite{Markov}.
Moreover matrix products play an essential role in the representation of
various computational processes, i.e., linear recurrent sequences \cite{tHaHaHiKa05a,OW_ICALP2015-1,OW_ICALP2015-2}, arithmetic circuits \cite{GOW_STACS2015}, 
hybrid and dynamical systems \cite{OSW2015, BP2008},
probabilistic and quantum automata \cite{Blondel2005}, stochastic
games, broadcast protocols \cite{EFM1999},  optical systems, etc.  New algorithms for solving reachability problems in matrix semigroups 
can be incorporated into software verification tools and used for analysis of 
mathematical models in physics, chemistry, biology, ecology, and economics.

However, many computational problems for matrix semigroups  are inherently difficult to solve
even when the problems are considered in dimension two, and most of these problems 
become undecidable in general starting from dimension three or four.
Examples of such problems are the Membership problem (including the special cases of the Mortality and Identity problems), vector reachability, scalar reachability, freeness problem and the emptiness problem of matrix semigroups intersection \cite{BP2012}. 
All above problems are tightly connected, including three central problems: 

\begin{itemize}
\item  {\bf The membership problem:} 
Let $S=\langle G\rangle$ be a 
semigroup generated by a finite set $G$ of ${n \times n}$ matrices. 
Determine whether a given matrix $M$ belongs to $S$, that is, 
determine whether there exists a sequence of matrices ${M_1,M_2, \ldots ,M_k}$ in $G$
such that $M={M_1\cdot M_2 \cdot \ldots \cdot M_k}$

\item
{\bf The vector reachability problem:}
Let $\mathbf{x}$ and $\mathbf{y}$ be two vectors and $S$ be a given finitely generated 
semigroup of ${n \times n}$ matrices.
Determine whether there is a matrix $M \in S$ such that $M\mathbf{x} = \mathbf{y}$.

\item
{\bf The scalar reachability problem:}
Let $\mathbf{x}$ and $\mathbf{y}$ be two vectors, $\lambda$ be a scalar, and $S$ be a given finitely generated 
semigroup of ${n \times n}$ matrices.
Determine whether there is a matrix $M \in S$ such that $\mathbf{x}^\top M\mathbf{y} = \lambda$.
\end{itemize}

The vector reachability problem can be seen as a parameterized version of the membership problem, where
some elements of a matrix $M$  are either independent variables or variables linked by some equations.
In contrast to the original membership problem, where all values of $M$ are
constants, in vector reachability we may have an infinite set of matrices that can transform a vector
$\mathbf{x}$ to $\mathbf{y}$. Thus the decidability results for the membership cannot be directly applied
to the vector reachability problem.

The scalar reachability can be viewed as a vector to hyperplane reachability problem. Indeed, we can rewrite the equation $\mathbf{x}^\top M\mathbf{y} = \lambda$ as a system of two equations: $M\mathbf{y} = \mathbf{z}$ and $\mathbf{x}^\top \mathbf{z} = \lambda$. So, the question becomes if there is a matrix $M\in S$ that maps a given vector $\mathbf{y}$ to a vector $\mathbf{z}$ that lies on a hyperplane $\mathbf{x}^\top \mathbf{z} = \lambda$. Because there are infinitely many vectors on a hyperplane, decidability of the scalar reachability problem does not follow directly from the decidability of the vector reachability problem.

Most of the problems such as membership, vector
reachability and freeness are undecidable for $3 \times 3$ integer matrices.
The undecidability proofs in matrix semigroups are mainly based on 
various techniques and methods of embedding
universal computations into three and four dimensional matrices and their products.
The case of dimension two is the most intriguing one
since there is some evidence that if these problems are undecidable,
then this cannot be proved using a construction similar to the
one used for dimensions 3 and 4. In particular, there is no injective
semigroup morphism from pairs of words over any finite
alphabet (with at least two elements) into complex $2 \times 2$ matrices
\cite{CHK99}, which means that the coding of independent pairs of words
in $2 \times 2$ complex matrices is impossible and the exact encoding of
the Post Correspondence Problem or a computation of a Turing
Machine cannot be used directly for proving undecidability
in $2 \times 2$ matrix semigroups over $\mathbb{Z}$, $\mathbb{Q}$ or $\mathbb{C}$.
The only undecidability result in dimension two for the vector reachability and the membership
problems has been shown in the case of  $2 \times 2$ matrices
over quaternions~\cite{BP_IC2008}.

The main hypothesis is that problems for ${2 \times 2}$ matrix semigroups 
over integers, rationals or complex  numbers could be decidable,
but not much is known about the status of these problems.
Recently, there was some progress on the {\sl Membership problem}, which was shown to be decidable
in $\SL$, and the {\sl Identity problem}, which was shown to be decidable in $\Z^{2\times 2}$~\cite{CK2005}. 
Later the decidability of the {\sl Freeness  problem} 
(that is, to decide whether each element can be expressed uniquely as a
product of generating matrices) was shown for $\SL$ \cite{CN2012}. 
On the other hand, the Mortality, Identity and vector reachability problems were
 shown to be at least NP-hard for $\SL$ in
\cite{BHP2012, BP2012}, but for the modular group
the membership was shown to be decidable in polynomial time 
by Gurevich and Schupp \cite{Gurevich2007}.

The algorithmic properties of $\SL$ are important in the context of many fundamental problems
in hyperbolic geometry \cite{Zagier2008,Chamizo,Elstrodt}, dynamical systems \cite{PR2004},
Lorenz/modular knots \cite{Mackenzie_2009}, braid groups \cite{Potapov13},
particle physics, high energy physics \cite{EDWARD_WITTEN},
M/string theories \cite{GM2011}, ray tracing analysis, music theory \cite{Noll}
and can lead to further decidability results in $\Z^{2\times 2}$ using matrix presentation
in the Smith normal form. 

This paper solves two open problems about the decidability of the vector reachability problem 
for finitely generated semigroups of matrices from $\SL$ and the point to point reachability (over rational numbers) 
for fractional linear transformations $f_M(x) = \frac{ax+b}{cx+d}\,$, where the associated matrix
 $M={\begin{bmatrix} a & b\\ c & d \end{bmatrix}}$ belongs to $\SL$.
The approach to solving these reachability problems for $2 \times 2$ matrix semigroups 
is based on the analysis of reachability paths between vectors or points. This analysis is then used to
translate the numerical reachability problems into computational
problems on words and regular languages. 
We also present several extensions of our main results, give a geometric interpretation of reachability paths, and use this technique to solve a special case of the scalar reachability problem.

The decidability proof of the vector reachability problem in dimension two 
presented in this paper is the first nontrivial new result 
for solving vector reachability problems since 1996 when it was shown that the problem 
is decidable for any commutative matrix semigroup in any dimension \cite{Babai} and for a special case of non-commuting matices \cite{LP2004}. On the other hand, in the general case of non-commuting matrices 
the problem is known to be undecidable already for integer matrices in dimension three \cite{HHH2007}.

\section{Preliminaries}

The integers and rationals are denoted by $\Z$ and $\Q$, respectively, and $\SL$ is a group of $2\times 2$ integer matrices with determinant~1. The notation $a \mid b$ means that $a$ divides $b$, and $a\nmid b$ means that $a$ does not divide $b$, when $a$ and $b$ are integer numbers.

\begin{definition}
With each matrix $M=\begin{bmatrix} a & b\\ c & d \end{bmatrix}\in \SL$ we associate a fractional linear map (also called M\"obius transformation) $f_M:\Q\to \Q$ defined as
$
f_M(x) = \frac{ax+b}{cx+d}\,.
$
This definition can be extended from $\Q$ to $\Q\cup \{\infty\}$ in a natural way by setting $f_M(\infty)=\frac{a}{c}$ if $c\neq 0$, $f_M(\infty)=\infty$ if $c=0$, and $f_M(x)=\infty$ if $cx+d=0$.

Note that we have $f_{M_1}\circ f_{M_2} = f_{M_1M_2}$ for any matrices $M_1$ and $M_2$.
\end{definition}

Let $M_1,\ldots,M_n$ be a finite collection of matrices. Then $\langle M_1,\ldots,M_n\rangle$ denotes the multiplicative semigroup (including the identity matrix) generated by $M_1,\ldots,M_n$.

\begin{definition}
The \emph{vector reachability problem} in $\SL$ is defined as follows: Given two vectors $\mathbf{x}$ and $\mathbf{y}$ with integer coefficients and a finite collection of matrices $M_1,\ldots,M_n$ from $\SL$, decide whether there exists a matrix $M\in \langle M_1,\ldots,M_n\rangle$ such that $M\mathbf{x}=\mathbf{y}$.
\end{definition}

\begin{definition}
The \emph{reachability problem by fractional linear transformations} in $\SL$ is defined as follows: Given two rational numbers $x$ and $y$ and a finite collection of matrices $M_1,\ldots,M_n$ from $\SL$, decide whether there exists a matrix $M\in \langle M_1,\ldots,M_n\rangle$ such that $f_M(x)=y$.
\end{definition}

The main result of our paper is that the vector reachability problem and the reachability problem by fractional linear transformations for $\SL$ are decidable (Theorem \ref{vrpthm}). Both proofs follow the same pattern. We will use the fact that any matrix $M$ from $\SL$ can be expressed as product of matrices $S=\begin{bmatrix} 0 & -1\\ 1 & 0 \end{bmatrix}$ and $R=\begin{bmatrix} 0 & -1\\ 1 & 1 \end{bmatrix}$. So we can represent any $M\in\SL$ by a word $w$ in the alphabet $\Sigma=\{S,R\}$.

The main idea of both proofs is to show that the solution set of the equation $M\mathbf{x}=\mathbf{y}$ or $f_M(x)=y$ has a form
$
\bigg\{B\begin{bmatrix} 1 & 1\\ 0 & 1 \end{bmatrix}^tC\ :\ t\in \Z\bigg\}, \text{ where } B,C \text{ are matrices from } \SL,
$
or a union of two such set, see Theorems \ref{regmatthm} and \ref{mobiusthm}, respectively. After translating matrices into words, these sets become regular languages. On the other hand, the language that corresponds to the semigroup $\langle M_1,\ldots,M_n\rangle$ is also regular. Indeed, if $M_i$ corresponds to the word $w_i$, then the semigroup $\langle M_1,\ldots,M_n\rangle$ translates into the language ${(w_1+\cdots+w_n)}^*$. The last step of the proof is to show that the emptiness problem of the intersection of two such languages is decidable (Proposition \ref{autom}).

Here is a more detailed description of our proofs. Let $M=\begin{bmatrix} a & b\\ c & d\end{bmatrix}$, $\mathbf{x}=\begin{bmatrix}x_1\\ x_2\end{bmatrix}$ and $\mathbf{y}=\begin{bmatrix}y_1\\ y_2\end{bmatrix}$. To show that the equation $M\mathbf{x}=\mathbf{y}$ defines a regular language we must solve the following system of three equations in four unknown variables:
\begin{align*}
 x_1a+x_2b = y_1&&
 x_1c+x_2d = y_2&&
 ad-bc = 1
\end{align*}

Choosing $b$ as a free parameter, we can reduce it to the following system of linear congruence equations:
\begin{align*}
  x_2b \equiv y_1\ (\bmod\ x_1)&&
  y_2b \equiv -x_1\ (\bmod\ y_1)&&
  x_2y_2b \equiv y_1y_2-x_1x_2\ (\bmod\ x_1y_1)
\end{align*}
By Lemma \ref{linsys} from Section \ref{sec:linmat} of the Appendix, the above system either has no solutions or it has a solution of the form $b=b_1t+b_2$, where $t\in \Z$, and hence all coefficients of the matrix $M$ are linear functions of $t$. In Proposition~\ref{regmat} we will show that such matrices can be written in the form $M=B\begin{bmatrix} 1 & k\\ 0 & 1 \end{bmatrix}^{t}C$, where $B$, $C$ are some matrices from $\SL$, $k$ is a fixed integer number and $t\in\Z$ is a free parameter. After that it is not hard to see that such solution translates into a regular language.

We will use a similar approach to prove that the equation $f_M(x)=y$ also defines a regular language. In fact, we will do it by showing that the solution set of  $f_M(x)=y$ is equal to the union of the solution sets of the equations $M\mathbf{x}=\mathbf{y}$ and $M\mathbf{x}= -\mathbf{y}$ for suitable vectors $\mathbf{x}$ and $\mathbf{y}$.

The final step is to show that there is an algorithm that decides whether the intersection of two regular subsets of $\SL$ is empty or not. Our idea relies on the fact that the intersection of two regular languages is regular, and that the emptiness problem for regular languages is decidable. The problem here is that we cannot apply these facts directly because for each matrix $M\in \SL$ there are infinitely many words ${w\in\{S,R\}}^*$ that correspond to $M$, and only some of them may appear in the given language. However there is only one \emph{reduced} word that corresponds to $M$, that is, the word that does not have a substring of the form $SS$ or $RRR$.
So, our solution is to take any automaton $\A$ and turn it into a new automaton $\wt{\A}$ that accepts the same language as $\A$ plus all reduced words $w$ that correspond to non-reduced words $w'$ accepted by $\A$.

Note that in $\SL$ we have an equality $S^2=R^3=-I$. Thus to construct $\wt{\A}$ we add to $\A$ a new $\varepsilon$-transition from a state $q_1$ to a state $q_2$ if there is a run of $\A$ from $q_1$ to $q_2$ labelled by $SS$ or $RRR$. We will apply this procedure iteratively until no new $\varepsilon$-transitions can be added. However we need to keep track of sign changes when we add new $\varepsilon$-transitions. To achieve this we will use \emph{signed automata}, which are slight modifications of the usual finite automata but they take into account such sign changes.

Now to solve the emptiness problem for the intersection of two regular languages $L_1$ and $L_2$, we take the signed automata $\A_1$ and $\A_2$ that accept $L_1$ and $L_2$, respectively, and construct new automata $\wt{\A}_1$ and $\wt{\A}_2$ as described above. After that we can check whether $L(\wt{\A}_1)\cap L(\wt{\A}_2)\neq\emptyset$.

In the Section \ref{sec:geom} we will show how to extend these
decidability results to arbitrary regular subsets of $\SL$, i.e., subsets that are defined by finite automata. Using this technique we will show how to algorithmically solve the equation
$M_1^{x_1}\cdots M_k^{x_k}\mathbf{x}=N_1^{y_1}\cdots N_l^{y_l}\mathbf{y}$, where $\mathbf{x},\mathbf{y}$ are given vectors from $\Z\times \Z$, the matrices $M_1,\dots, M_k$ and $N_1,\dots, N_l$ are from $\SL$, and $x_1,\dots,x_k$ and $y_1,\dots,y_l$ are unknown non-negative integers. Furthermore, we will show how to apply this method to prove that a special case of the scalar reachability problem is decidable.

\section{Main results}\label{main}

The characterization of the solution set of the equation $M\mathbf{x}=\mathbf{y}$ given in Theorem \ref{regmatthm} will follow from Propositions \ref{linmat} and \ref{regmat}. First, we prove one simple lemma which we will use several times in our arguments.
\begin{lemma} \label{lem:gcd}
Let $\mathbf{x}=\begin{bmatrix}x_1\\ x_2\end{bmatrix}$ and $\mathbf{y}=\begin{bmatrix}y_1\\ y_2\end{bmatrix}$ be vectors from $\Z\times\Z$ and $M$ be a matrix from $\SL$ such that $M\mathbf{x}=\mathbf{y}$. Then $\gcd(x_1,x_2)=\gcd(y_1,y_2)$.
\end{lemma}

\begin{proof}
Take any $k\in\Z$ such that ${k\mid x_1,x_2}$ and let $M=\begin{bmatrix} a & b\\ c & d \end{bmatrix}$. Then from $M\mathbf{x}=\mathbf{y}$ we have $y_1=ax_1+bx_2$ and $y_2=cx_1+dx_2$. Thus $k\mid y_1,y_2$. Now since $M\in\SL$, $M^{-1}$ is also in $\SL$, and $M\mathbf{x}=\mathbf{y}$ is equivalent to $M^{-1}\mathbf{y}=\mathbf{x}$. So, if $k\in \Z$ is any number such that $k\mid y_1,y_2$, then $k\mid x_1,x_2$. Therefore, $\gcd(x_1,x_2)=\gcd(y_1,y_2)$.
\end{proof}

\setcounter{thm1}{\value{theorem}}
\begin{proposition} \label{linmat}
  Let $\mathbf{x}=\begin{bmatrix}x_1\\ x_2\end{bmatrix}$ and $\mathbf{y}=\begin{bmatrix}y_1\\ y_2\end{bmatrix}$ be two vectors from $\Z\times\Z$, such that $\mathbf{x}$ is not equal to the zero vector\/ $\mathbf{0}$, and consider the matrix equation $M\mathbf{x}=\mathbf{y}$, where $M$ is an unknown matrix from $\SL$. Then either this equation does not have a solution or all its solutions are given by $M=tA_1+A_2$, where $t$ is any integer number, $A_1,A_2$ are some matrices from $\Z^{2\times 2}$ such that $A_1$ is a nonzero matrix. Moreover, there is a polynomial time algorithm that determines whether such equation has a solution and if so, finds it.
\end{proposition}

\begin{proof}
See Section \ref{sec:linmat} of the Appendix.
\end{proof}

For the next proposition we will need the following theorem about the Smith normal form of a matrix.
\begin{theorem}[Smith normal form \cite{KB79}] \label{SNF}
For any nonzero matrix $A\in \Z^{2\times 2}$, there are matrices $B,C$ from $\SL$ such that
$
  A=B\begin{bmatrix} t_1 & 0\\ 0 & t_2 \end{bmatrix}C
$
for some $t_1,t_2\in \Z$ such that $t_1\neq 0$ and $t_1 \mid t_2$. Moreover, $B$, $C$, $t_1$, $t_2$ can be computed in polynomial time.
\end{theorem}

\begin{proposition} \label{regmat}
Let $A_1$ and $A_2$ be matrices from $\Z^{2\times 2}$ such that $A_1$ is a nonzero matrix and, for every $t\in \Z$, we have $tA_1+A_2\in \SL$. Then there are matrices $B$ and $C$ from $\SL$ and $k\in \Z$ such that\\
\null\qquad$
  tA_1+A_2 = BT^{kt}C\ \text{ for every } t\in \Z,
$\\
where $T = \begin{bmatrix} 1 & 1\\ 0 & 1 \end{bmatrix}\in \SL$. Moreover, $B$, $C$, and $k$ can be computed in polynomial time.
\end{proposition}

\begin{proof}
Let $A_1=\begin{bmatrix} a_1 & b_1\\ c_1 & d_1\end{bmatrix}$ and $A_2=\begin{bmatrix} a_2 & b_2\\ c_2 & d_2\end{bmatrix}$. By the assumption, for every $t\in \Z$,
$
  \begin{vmatrix}
   a_1t+a_2 & b_1t+b_2\\ 
   c_1t+c_2 & d_1t+d_2\\ 
  \end{vmatrix} = 1.
$
That is
$(a_1t+a_2)(d_1t+d_2) - (b_1t+b_2)(c_1t+c_2) =1$ or\\
\null\qquad$(a_1d_1-b_1c_1)t^2 + (a_1d_2+a_2d_1-b_1c_2-b_2c_1)t + a_2d_2-b_2c_2 =1$ for all $t\in \Z$.\\
Therefore, $a_1d_1-b_1c_1=0$, $a_1d_2+a_2d_1-b_1c_2-b_2c_1=0$, and $a_2d_2-b_2c_2=1$. In particular, $\det(A_1)=0$ and $\det(A_2)=1$.

By Theorem \ref{SNF}, there are matrices $F,G\in \SL$ such that $A_1=F\begin{bmatrix} k & 0\\ 0 & l \end{bmatrix}G$ for some $k,l\in \Z$ such that $k\mid l$. Since $\det(A_1)=0$ we have that $kl=0$. However if $k=0$ and $l=0$, then $A_1$ is equal to the zero matrix, contrary to the assumption. Hence we must have that $k\neq 0$ and $l=0$.

Now $F^{-1}(tA_1+A_2)G^{-1} = \begin{bmatrix} kt+a & b\\ c & d \end{bmatrix}$, for some $a,b,c,d \in \Z$. Note that since $\det(F)=\det(G)=\det(tA_1+A_2)=1$, we have
$\begin{vmatrix}
 kt+a & b\\ 
 c & d\\ 
\end{vmatrix} = dkt+ad-bc = 1$ for every $t\in \Z$.
Hence $dk=0$ and so $d=0$. Substituting $d=0$ in the above equation, we obtain $bc=-1$. Since $b$ and $c$ are integers, there are only two possibilities: $b=1$, $c=-1$, or $b=-1$, $c=1$. So the above matrix actually looks like
$
F^{-1}(tA_1+A_2)G^{-1} = \begin{bmatrix} kt+a & \mp1\\ \pm1 & 0 \end{bmatrix}.
$
Therefore, $T^{-c(kt+a)}F^{-1}(tA_1+A_2)G^{-1} = D$, where $c=\pm1$ and $D= \begin{bmatrix} 0 & \mp1\\ \pm1 & 0 \end{bmatrix}\in \SL$. Hence
$
  tA_1+A_2 = FT^{(ck)t}T^{ca}DG.
$
Note that $F$ and $T^{ca}DG$ are in $\SL$. This completes the proof. The bound on complexity follows from the fact that $F$ and $G$ can be computed in PTIME by Theorem \ref{SNF}.
\end{proof}

As a corollary of Propositions \ref{linmat} and \ref{regmat} we obtain the following theorem.

\begin{theorem} \label{regmatthm}
Let $\mathbf{x}=\begin{bmatrix}x_1\\ x_2\end{bmatrix}$ and $\mathbf{y}=\begin{bmatrix}y_1\\ y_2\end{bmatrix}$ be vectors from $\Z\times\Z$ such that $\mathbf{x}\neq \mathbf{0}$, and consider the matrix equation $M\mathbf{x}=\mathbf{y}$, where $M$ is an unknown matrix from $\SL$. Then either this equation does not have a solution or all its solutions are given by the following formula
$M = B{\begin{bmatrix} 1 & k\\ 0 & 1 \end{bmatrix}}^{t}C$, where $t\in \Z$.

In the above expression $B$ and $C$ are some matrices from $\SL$, and $k$ is an integer number. Moreover, there is a polynomial time algorithm that determines whether such equation has a solution and if so, finds the suitable matrices $B$, $C$ and the integer $k$.
\end{theorem}

In Section \ref{sec:geom} we will give a geometric interpretation  of reachability paths (Figure \ref{fig:lin} and Proposition \ref{prop:uv}), using which we can prove the following corollary.\footnote{Even though we use Corollary \ref{regmatcor} in the proofs of Theorem \ref{mobiusthm} and Proposition \ref{srp}, it is not essential there for proving decidability. Namely, all references to Corollary \ref{regmatcor} in these proofs can be replaced by references to Theorem~\ref{regmatthm}, at the same time replacing $T$ with $T^k$ where appropriate.} The proof itself can be found in Section \ref{sec:cor} of the Appendix.

\setcounter{thm4}{\value{theorem}}
\begin{corollary} \label{regmatcor}
The value of the parameter $k$ in Theorem~\ref{regmatthm} is equal to 1.
\end{corollary}

Theorem \ref{regmatthm} provides us with a characterization of the matrices $M\in \SL$ that map vector $\mathbf{x}$ to vector $\mathbf{y}$. This characterization will be used later to prove the decidability of the vector reachability problem. We now give a similar characterization of the matrices $M\in \SL$ for which the fractional linear transformation $f_M$ maps a number $x$ to number $y$. In fact, we will do this by reducing the problem to finding the solutions of the equation $M\mathbf{x}=\mathbf{y}$ which we discussed above.

\begin{theorem}\label{mobiusthm}
Let $x$ and $y$ be rational numbers and let $\F(x,y)$ be the following set of matrices from $\SL${\rm :}\\
\null\qquad
$
\F(x,y)=\{M\in \SL\ :\ f_M(x)=y\}.
$\\
Then $\F(x,y) = \F_1(x,y)\cup \F_2(x,y)$, where each $\F_i(x,y)$ is either empty or has the form\\
\null\qquad
$
\F_i(x,y) = \{B_iT^tC_i\ :\ t\in \Z\},
$\\
where $B_i$ and $C_i$ are some matrices from $\SL$. Moreover, there is a polynomial time algorithm that determines whether each $\F_i(x,y)$ is empty or not and in the latter case finds corresponding matrices $B_i$ and $C_i$. 
\end{theorem}

\begin{proof}
Let us write the numbers $x$ and $y$ as $x=\frac{x_1}{x_2}$ and $y=\frac{y_1}{y_2}$, where we assume that $\gcd(x_1,x_2) = \gcd(y_1,y_2) = 1$. Consider the equation $f_M(x)=y$, where $M=\begin{bmatrix} a & b\\ c & d \end{bmatrix}$ is an unknown matrix from $\SL$. We can rewrite it as
\begin{equation} \label{eq:frac}
\frac{a\frac{x_1}{x_2}+b}{c\frac{x_1}{x_2}+d}=\frac{y_1}{y_2}\quad \text{or}\quad
\frac{ax_1+bx_2}{cx_1+dx_2}=\frac{y_1}{y_2}\,.
\end{equation}
Consider the vectors $\mathbf{x}=\begin{bmatrix}x_1\\ x_2\end{bmatrix}$, $\mathbf{y}=\begin{bmatrix}y_1\\ y_2\end{bmatrix}$, and $\mathbf{z}=\begin{bmatrix}z_1\\ z_2\end{bmatrix}$, where $\mathbf{z}$ is the vector with coordinates $z_1=ax_1+bx_2$ and $z_2=cx_1+dx_2$. So we have that $\mathbf{z}=M\mathbf{x}$. In this notation Equation (\ref{eq:frac}) is equivalent to the fact that vector $\mathbf{z}=M\mathbf{x}$ belongs to the set $\{k\mathbf{y}\ :\ k\in \Z\}$.

Recall that $\gcd(x_1,x_2)=1$ and hence, by Lemma \ref{lem:gcd}, we also have that $\gcd(z_1,z_2)=1$. Thus if $\mathbf{z}=k\mathbf{y}$ for some $k\in\Z$, then we must have that $k=\pm 1$. In other words, we showed that Equation~(\ref{eq:frac}) is equivalent to two matrix equations: $M\mathbf{x}=\mathbf{y}$ and $M\mathbf{x}=-\mathbf{y}$. So we have that
$\F(x,y) = \F_1(x,y)\cup \F_2(x,y)$, where\\
\null\quad
$
\F_1(x,y) = \{M\in \SL\ :\ M\mathbf{x}=\mathbf{y}\}\quad \text{and}\quad
\F_2(x,y) = \{M\in \SL\ :\ M\mathbf{x}=-\mathbf{y}\}.
$\\
Note that $\mathbf{x}\neq \mathbf{0}$ because $x_2\neq 0$. Hence by Theorem \ref{regmatthm} and Corollary \ref{regmatcor}, each $\F_i(x,y)$ is either empty or has the form $\F_i(x,y) = \{B_iT^tC_i\ :\ t\in \Z\}$ for some $B_i$ and $C_i$ from $\SL$ which can be computed in polynomial time. 
\end{proof}

Now we will use signed automata to prove that the emptiness problem for the intersection of two regular subsets of $\SL$ is decidable.

Consider an alphabet $\Sigma=\{S,R\}$ consisting of two symbols $S$ and $R$ and define the mapping $\phi:\Sigma\to \SL$ as follows: $\phi(S)=\begin{bmatrix} 0 & -1\\ 1 & 0 \end{bmatrix}$ and $\phi(R)=\begin{bmatrix} 0 & -1\\ 1 & 1 \end{bmatrix}$. We can extend this mapping to the morphism $\phi:\Sigma^*\to \SL$ in a natural way. The matrices $\phi(S)$ and $\phi(R)$ are in fact generators of $\SL$, so $\phi$ is surjective.
We call a word $w\in \Sigma^*$ \emph{reduced} if it does not have substrings of the form $SS$ or $RRR$. In our proof we will make use of the following well-known fact.

\begin{theorem}[\cite{LS,MKS,Ran}] \label{thm:uniq}
For every $M\in \SL$, there exists a unique reduced word $w\in \Sigma^*$ such that either $M=\phi(w)$ or $M=-\phi(w)$.
\end{theorem}

\begin{definition}
A \emph{signed automaton} $\A=(\Sigma,Q,I,\Delta,F^+,F^-)$ is a (non-deterministic) finite automaton  whose final states are divided into two (not necessarily  disjoint) subsets $F^+$ and $F^-$.

A \emph{signed language} accepted by a signed automaton $\A$ is a pair $L(\A)=(L(\A)^+,L(\A)^-)$, where $L(\A)^+$ and $L(\A)^-$ consists of the words $w\in \Sigma^*$ for which there is a run of $\A$ that ends in the set $F^+$ or $F^-$, respectively. Note that we do not assume that $L(\A)^+$ and $L(\A)^-$ are disjoint.
\end{definition}

Let $L=(L^+,L^-)$ be a signed language, then we define a regular subset of $\SL$ corresponding to this language as
$
\phi(L) = \{\phi(w)\,:\, w\in L^+\}\cup \{-\phi(w)\, :\, w\in L^-\}.
$

The following proposition is an important ingredient of our main results.

\setcounter{thm2}{\value{theorem}}
\begin{proposition}\label{autom}
There is an algorithm that for any given regular signed languages $L_1$ and~$L_2$ over the alphabet $\Sigma$, decides whether $\phi(L_1)\cap\phi(L_2)$ is empty or not. 
\end{proposition}

\begin{proof}
See Section \ref{sec:autom} of the Appendix.
\end{proof}

We are now ready to prove our main results.

\begin{theorem}\label{vrpthm}
The vector reachability problem and the reachability problem by fractional linear transformations in $\SL$ are decidable.
\end{theorem}

\begin{proof}
Suppose $M_1,\ldots,M_n$ is a given finite collection of matrices from $\SL$. Let $w_1,\ldots,w_n\in\Sigma^*$ be some words, not necessarily reduced, such that $M_i=\phi(w_i)$, for $i=1,\ldots,n$. Define the language $\LL_{\mathit{semigr}}$ that corresponds to the semigroup $\langle M_1,\ldots,M_n\rangle$ as
$
\LL_{\mathit{semigr}} = {(w_1+w_2+\cdots+w_n)}^*.
$

Recall that in the vector reachability problem we are given two vectors $\mathbf{x}$ and $\mathbf{y}$ from $\Z\times\Z$, and we ask if there is a matrix $M\in \langle M_1,\ldots,M_n\rangle$ such that $M\mathbf{x}=\mathbf{y}$. We want to construct a regular language $\LL^{\mathrm{vrp}}_{\mathbf{x},\mathbf{y}}$ that corresponds to these matrices.

If $\mathbf{x}=\mathbf{0}$ and $\mathbf{y}\neq\mathbf{0}$, then we set $\LL^{\mathrm{vrp}}_{\mathbf{x},\mathbf{y}}=\emptyset$ because in this case the equation $M\mathbf{x}=\mathbf{y}$ does not have a solution. On the other hand, if $\mathbf{x}=\mathbf{0}$ and $\mathbf{y}=\mathbf{0}$, then we set $\LL^{\mathrm{vrp}}_{\mathbf{x},\mathbf{y}}={\{S,R\}}^*$ because any matrix $M\in \SL$ satisfies the equation $M\mathbf{0}=\mathbf{0}$.

Now assume that $\mathbf{x}\neq\mathbf{0}$. Then by Theorem \ref{regmatthm}, the matrix equation $M\mathbf{x}=\mathbf{y}$ either has no solution, or its solution has the form $\{BT^{t}C : t\in \Z\}$, where $T=\begin{bmatrix} 1 & 1\\ 0 & 1 \end{bmatrix}$, and $B$ and $C$ are some matrices from $\SL$. Moreover, $B$ and $C$ can be computed from $\mathbf{x}$ and $\mathbf{y}$ in PTIME.
In the case when $M\mathbf{x}=\mathbf{y}$ has no solution, we set $\LL^{\mathrm{vrp}}_{\mathbf{x},\mathbf{y}}=\emptyset$. If the solution set in non-empty, then we can rewrite it as\\
\null\qquad
$
\{BT^{t}C\, :\, t\in \Z\} = \{BT^tC\, :\, t\geq 0\} \cup \{BT^{-t}C\, :\, t\geq 0\}.
$\\
Let $u$ and $v$ be words from $\Sigma^*$ such that $B=\phi(u)$ and $C=\phi(v)$. It is easy to check that $T=\phi(S^3R)$ and $T^{-1}=\phi(R^5S)$. Hence
$
\LL^{\mathrm{vrp}}_{\mathbf{x},\mathbf{y}} = u{(S^3R)}^*v+u{(R^5S)}^*v
$
is a regular language that describes the solutions of the equation $M\mathbf{x}=\mathbf{y}$ in $\SL$.

In a similar way we can construct a regular language $\LL^{\mathrm{flt}}_{x,y}$ that corresponds to the reachability problem by fractional linear transformations from $x$ to $y$. By Theorem \ref{mobiusthm}, the set $\F(x,y)$ of matrices from $\SL$ that satisfy the equation $f_M(x)=y$ is equal to $\F(x,y) = \F_1(x,y)\cup \F_2(x,y)$, where each $\F_i(x,y)$ is either empty or has the form
$\F_i(x,y) = \{B_iT^tC_i : t\in \Z\}$, where $T$ is as above, and $B_i$ and $C_i$ are some matrices from $\SL$. All these matrices can be computed in PTIME from $x$ and $y$.

We define $\LL^{\mathrm{flt}}_{x,y}$ as the union $\LL^{\mathrm{flt}}_{x,y} = \LL_1\cup \LL_2$ of two regular languages $\LL_1$ and $\LL_2$.  If $\F_i(x,y)$ is empty, then we set $\LL_i=\emptyset$. Otherwise, let $u_i$ and $v_i$ be words from $\Sigma^*$ such that $B_i=\phi(u_i)$ and $C_i=\phi(v_i)$. Then we can define $\LL_i$ as
$
\LL_i = u_i{(S^3R)}^*v_i+u_i{(R^5S)}^*v_i.
$
Thus we defined a regular language $\LL^{\mathrm{flt}}_{x,y}$ that corresponds the solution set of the equation $f_M(x)=y$ in $\SL$.

We remind that in Proposition \ref{autom} we work with signed languages. Therefore, in what follows we convert every regular language $L$ that we have constructed so far into a corresponding signed language $(L,\emptyset)$.

Finally, the vector reachability problem for $\mathbf{x}$ and $\mathbf{y}$ has a solution if and only if\\
\null\qquad
$
\phi\big((\LL^{\mathrm{vrp}}_{\mathbf{x},\mathbf{y}},\emptyset)\big)\cap \phi\big((\LL_{\mathit{semigr}},\emptyset)\big)\neq \emptyset.
$\\
Similarly, the reachability problem by fractional linear transformations for $x$ and $y$ has a solution if and only if\\
\null\qquad
$
\phi\big((\LL^{\mathrm{flt}}_{x,y},\emptyset)\big)\cap \phi\big((\LL_{\mathit{semigr}},\emptyset)\big)\neq \emptyset.
$\\
By Proposition \ref{autom} these questions are algorithmically decidable.
\end{proof}

A characterization of the matrices $M$ from $\SL$ that satisfy the equation $M\mathbf{x}=\mathbf{y}$, which is given in Theorem \ref{regmatthm}, can be computed in polynomial time. However the overall complexity of the algorithm is EXPTIME due to the 
fact that a reduced word $w$ that corresponds to a given matrix $M$, i.e., such that $M=\pm\phi(w)$, has length exponential in the decimal presentation of $M$. So computing symbolic presentations of given matrices and constructing an automaton for the language $\LL_{\mathit{semigr}}$ takes exponential time. The next steps of the algorithm take only polynomial time in the size of this automaton.
However the PTIME algorithm for computing all mappings from $\mathbf{x}$ to $\mathbf{y}$
could be combined with the result of Gurevich and Schupp \cite{Gurevich2007}
to produce a polynomial time algorithm for the vector reachability problem over the modular group. 
Moreover, any improvement of EXPTIME solution proposed in \cite{CK2005} will
improve the complexity of the vector reachability problem.

\section{Geometric interpretation and extensions}\label{sec:geom}

Consider a semigroup generated by matrices $M_1,\ldots,M_n$ from $\SL$. As we showed above, this semigroup can be described by a regular language which we called $\LL_{\mathit{semigr}}$. It's not hard to see that the proof of Theorem \ref{vrpthm} remains valid if we replace $\LL_{\mathit{semigr}}$ by any other regular language, that is, a language defined by a finite automaton or a labelled transition system.

\begin{proposition}\label{genreg}
Suppose that we are given a finite collection of matrices $M_1,\ldots,M_n$ from $\SL$ and a regular language $L\subseteq {\{1,\ldots,n\}}^*$. Consider the following generalized reachability problems:
\begin{itemize}
\item {\bf Generalized vector reachability problem.} Given two vectors $\mathbf{x}$ and $\mathbf{y}$ with integer coefficients, decide whether there exists a word $i_1\ldots i_k$ from the language $L$ such that $M_{i_1}\cdots M_{i_k}\mathbf{x}=\mathbf{y}$.

\item {\bf Generalized reachability problem by fractional linear transformations.} Given two rational numbers $x$ and $y$, decide whether there exists a word $i_1\ldots i_k$ from $L$ such that $f_{M_{i_1}\cdots M_{i_k}}(x)=y$.
\end{itemize}
Then the above generalized reachability problems are decidable.
\end{proposition}

\begin{proof}
The proof of this proposition is similar to the proof of Theorem \ref{vrpthm}. Namely, it follows from the fact that a regular language $L$ defines a regular subset in $\SL$ and Proposition~\ref{autom}, where we proved that the emptiness problem for the intersection of two regular subsets in $\SL$ is decidable.
\end{proof}

As an application of Proposition \ref{genreg} let us consider the follow matrix equation
\begin{equation}
\label{eq:mat}
M_1^{x_1}\cdots M_k^{x_k}\mathbf{x}=N_1^{y_1}\cdots N_l^{y_l}\mathbf{y},
\end{equation}
where $x_1,\dots,x_k$ and $y_1,\dots,y_l$ are non-negative integers. In \cite{Babai} it was proved that if $M_1,\dots,M_k$ and $N_1,\dots,N_l$ are commuting $n\times n$ matrices over algebraic numbers and $\mathbf{x},\mathbf{y}$ are vectors with algebraic coefficients, then it is decidable in polynomial time whether the Equation (\ref{eq:mat}) has a solution. On the other hand, in \cite{BHHKP08} it was shown that there is no algorithm for solving the equation $M_1^{x_1}\cdots M_k^{x_k}=Z$, where $M_1,\dots,M_k$ are integer $n\times n$ matrices and $Z$ is the zero matrix. Using the construction of Kronecker (or tensor) product of matrices, it is possible to show that the above-mentioned result implies that Equation (\ref{eq:mat}) is algorithmically undecidable in general for non-commuting integer matrices $M_1,\dots,M_k$ and $N_1,\dots,N_l$.

However with the help of Proposition \ref{genreg} we can algorithmically solve Equation (\ref{eq:mat}) in the case when $M_1,\dots,M_k$ and $N_1,\dots,N_l$ are matrices from $\SL$ and the vectors $\mathbf{x},\mathbf{y}$ have integer coefficients. Indeed, since the matrices from $\SL$ are invertible, we can rewrite (\ref{eq:mat}) as ${(N_l^{-1})}^{y_l}\cdots {(N_1^{-1})}^{y_1}M_1^{x_1}\cdots M_k^{x_k}\mathbf{x}=\mathbf{y}$. It is not hard to see that $\{{(N_l^{-1})}^{y_l}\cdots {(N_1^{-1})}^{y_1}M_1^{x_1}\cdots M_k^{x_k} : x_1,\dots,x_k,\, y_1,\dots,y_l\in \N\cup\{0\}\,\}$ is a regular subset of $\SL$, and hence the problem is decidable. Using the same idea we can algorithmically solve Equation (\ref{eq:mat}) also in the case when $x_1,\dots,x_k$ and $y_1,\dots,y_l$ are arbitrary integers and the matrices are from $\SL$.

In the rest of this section we will give a geometric interpretation of both reachability problems (Figure \ref{fig:lin}), which we will use later to solve a special case of the scalar reachability problem (Proposition \ref{srp}).

\begin{proposition}\label{prop:uv}
According to Theorem \ref{regmatthm}, the set of matrices $M$ from $\SL$ that transform a vector $\mathbf{x}=\begin{bmatrix}x_1\\ x_2\end{bmatrix}$ to a vector $\mathbf{y}=\begin{bmatrix}y_1\\ y_2\end{bmatrix}$ has the form $\F=\{BT^{kt}C\ :\ t\in \Z\}$.

Consider the equation $BT^{kt}C\mathbf{x}=\mathbf{y}$ and let us make the following change of variables: $\mathbf{u}=C\mathbf{x}$ and $\mathbf{v}=B^{-1}\mathbf{y}${\rm :}
\quad$\mathbf{x}\xrightarrow{\ C\ }\mathbf{u}\xrightarrow{\ T^{kt}\ }\mathbf{v}\xrightarrow{\ B\ }\mathbf{y}$.
Then $\mathbf{u}=\mathbf{v}=\begin{bmatrix}d\\ 0\end{bmatrix}$, where $|d|=\gcd(x_1,x_2)=\gcd(y_1,y_2)$.
\end{proposition}

\begin{proof}
In the new notations, the equation $BT^{kt}C\mathbf{x}=\mathbf{y}$ can be written as $T^{kt}\mathbf{u}=\mathbf{v}$, and this equality holds for any $t\in \Z$.  Now let $\mathbf{u}=\begin{bmatrix}u_1\\ u_2\end{bmatrix}$ and $\mathbf{v}=\begin{bmatrix}v_1\\ v_2\end{bmatrix}$. Hence we have $\begin{bmatrix} 1 & kt\\ 0 & 1 \end{bmatrix}\begin{bmatrix}u_1\\ u_2\end{bmatrix} = \begin{bmatrix}v_1\\ v_2\end{bmatrix}$, which is equivalent to $u_2=v_2$ and $u_1+ktu_2=v_1$, for any $t\in \Z$. So, we must have $u_2=v_2=0$ and hence $u_1=v_1$.

Therefore, the vectors $\mathbf{u}$ and $\mathbf{v}$ have the form
$\mathbf{u}=\mathbf{v}=\begin{bmatrix}d\\ 0\end{bmatrix}$ for some $d\in \Z$.
Moreover, since $\mathbf{u}=C\mathbf{x}$, we obtain from Lemma \ref{lem:gcd} that $|d|=\gcd(x_1,x_2)=\gcd(y_1,y_2)$.
\end{proof}

We can give the following geometric interpretation of the transformation $BT^{t}C\mathbf{x}=\mathbf{y}$: first, we apply $C$ to $\mathbf{x}$ and arrive at $\mathbf{u}=\begin{bmatrix}d\\ 0\end{bmatrix}$, then we loop at $\mathbf{u}$ for $t$ many times using $T$, and finally apply $B$ to move from $\mathbf{u}$ to $\mathbf{y}$ (see Figure \ref{fig:lin} on the left).
\begin{figure}[h]
\centering
\begin{tikzpicture}[shorten >=1pt,scale=1.3]
\tikzset{dots/.style={circle,draw,fill,inner sep=0pt,minimum size=4pt}}

\node (x) at (3,1) [dots,label=right:$\mathbf{x}$] {};
\node (y) at (-1,1.5) [dots,label=left:$\mathbf{y}$] {};
\node (z) at (1,0) [dots,label=above:${(d,0)}$] {};

\draw[->,semithick] (0,-1) -- (0,1.6);
\draw[->,semithick] (-1.3,0) -- (3.3,0);

\draw[->,semithick] (x) to [bend left=10] node [below,midway] {$C$} (z);
\draw[->,semithick] (z) to [bend left=15] node [auto,near end] {$B$} (y);
\draw[->,semithick] (z) .. controls (2,-1) and (0,-1) .. node [auto,pos=0.2] {$T^t$} (z);

\node at (8,0) (inf) {$\infty$};

\node (x) at (4.5,0) [dots,label=below:$x$] {};
\node (y) at (6,0) [dots,label=below:$y$] {};

\draw[->,semithick] (4,0) -- (7.5,0);

\draw[->,semithick] (inf) .. controls (7,1) and (9,1) .. node [above,midway] {$f_{T^t}$} (inf);
\draw[->,semithick] (x) to [bend left] node [above,midway] {$f_C$} (inf);
\draw[->,semithick] (inf) to [bend left] node [below,midway] {$f_B$} (y);
\end{tikzpicture}
\caption{Geometric interpretation of the linear transformation $\mathbf{y}=BT^tC\mathbf{x}$ (left) and of the fractional linear transformation $y=f_{BT^tC}(x)$ (right).}
\label{fig:lin}
\end{figure}
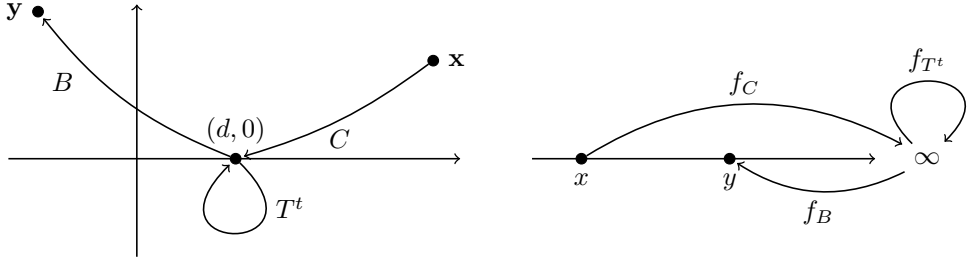

Similarly, we have the following geometric interpretation of the fractional linear transformation $y=f_{BT^tC}(x) = f_B\circ f_{T^t}\circ f_C(x)$: first it maps $x$ to $\infty$ using $f_C$, then loops at $\infty$ for $t$ many times using $f_T$, and finally maps $\infty$ to $y$ using $f_B$ (see Figure \ref{fig:lin} on the right).

\newpage
We now show how to apply the geometric interpretation of the vector reachability problem to solve a special case of the scalar reachability problem. 
\begin{definition}
The \emph{scalar reachability problem} in $\SL$ is stated as follows:
Let $[z_1, z_2]$ and $\begin{bmatrix}x_1\\ x_2\end{bmatrix}$ be vectors from $\Z\times \Z$ and let $\lambda$ be an integer number. We are also given a finite collection of matrices $M_1,\ldots,M_n$ from $\SL$. The question is to decide whether there exists a matrix $M\in \langle M_1,\ldots,M_n\rangle$ which satisfies the equation
$
[z_1, z_2]\, M\begin{bmatrix}x_1\\ x_2\end{bmatrix}=\lambda.
$
\end{definition}

We will consider a special case of this problem when $z_2=1$ and $\lambda=1$. Our proof relies on the characterization from Theorem \ref{regmatthm} and Corollary~\ref{regmatcor} and on Proposition \ref{autom} in which we showed that the emptiness problem for the intersection of two regular subsets in $\SL$ is decidable.

\begin{proposition}\label{srp}
Suppose that the above equation has the form
\begin{equation} \label{eq:sc}
[a,1]\, M\begin{bmatrix}x_1\\ x_2\end{bmatrix}=1,
\end{equation}
where $a$, $x_1$ and $x_2$ are some integer numbers. Then this special case of the scalar reachability problem is decidable.
\end{proposition}

\begin{proof}
The general idea of the proof is the same as in Theorem \ref{vrpthm}, that is, we will show that the set of matrices $M\in \SL$ that satisfy Equation (\ref{eq:sc}) can be described by a regular language.
First, let us consider a geometric interpretation of this problem. We can rewrite Equation~(\ref{eq:sc}) as a system of two equations:
$M\begin{bmatrix}x_1\\ x_2\end{bmatrix} = \begin{bmatrix}y_1\\ y_2\end{bmatrix}$ and
$ay_1+y_2 = 1 $.
So, $M$ satisfies Equation (\ref{eq:sc}) if and only if it maps a fixed vector $\mathbf{x}=\begin{bmatrix}x_1\\ x_2\end{bmatrix}$ to some vector $\mathbf{y}=\begin{bmatrix}y_1\\ y_2\end{bmatrix}$ that lies on the line $L$ described by the equation
$
ay_1+y_2 = 1.
$
In other words, we have a \emph{vector to line reachability problem} for the line $L$ that is defined by the equation $ay_1+y_2 = 1$.

Note that if a vector $\mathbf{y}$ lies of the line $ay_1+y_2 = 1$, then $\gcd(y_1,y_2)=1$. Hence by Lemma \ref{lem:gcd}, Equation (\ref{eq:sc}) has a solution only if $\gcd(x_1,x_2)=1$. So, from now on we assume that $\gcd(x_1,x_2)=1$.

By Corollary \ref{regmatcor}, any $M\in \SL$ that maps $\mathbf{x}$ to a vector $\mathbf{y}$ on the line $L$ has the form $M=BT^tC$, where $B$ and $C$ are some matrices from $\SL$ and $t\in \Z$. Geometrically, the transformation $\mathbf{y}=BT^tC\mathbf{x}$ goes via the point $(1,0)$ as shown in Figure~\ref{fig:scalar}.

Note that the matrices $B$ and $C$ above depend on the vector $\mathbf{y}$ as a parameter. Here we prove a useful lemma which will imply that we can choose only one matrix $C$ that maps $\mathbf{x}$ to $\begin{bmatrix}1\\ 0\end{bmatrix}$ independently of the vector $\mathbf{y}$.

\setcounter{thm3}{\value{theorem}}
\begin{lemma} \label{lem:ind}
Let $\mathbf{x}=\begin{bmatrix}x_1\\ x_2\end{bmatrix}$ and $\mathbf{y}=\begin{bmatrix}y_1\\ y_2\end{bmatrix}$ be any vectors from $\Z\times \Z$ such that $\gcd(x_1,x_2)=\gcd(y_1,y_2)=d$. Let $d_1$ and $d_2$ be any integer numbers with $|d_1|=|d_2|=d$ and let $A_1$, $B_1$ and $A_2$, $B_2$ by any matrices from $\SL$ such that $B_i\mathbf{x}=\begin{bmatrix}d_i\\ 0\end{bmatrix}$ and $A_i\begin{bmatrix}d_i\\ 0\end{bmatrix}= \mathbf{y}$, for $i=1,2$. Then
$
\{A_1T^tB_1\ :\ t\in \Z\} = \{A_2T^tB_2\ :\ t\in \Z\}.
$
In other words, the following diagrams define the same set of matrices that map $\mathbf{x}$ to $\mathbf{y}$.
\begin{center}
\begin{tikzpicture}[scale=1.3]
\node at (0,0) (x) {$\mathbf{x}$};
\node at (1,0) (u) {$\left[\!\begin{smallmatrix}d_1\\ 0\end{smallmatrix}\!\right]$};
\node at (2,0) (y) {$\mathbf{y}$};

\draw[->] (x) -- (u) node [above,midway] {$B_1$};
\draw[->] (u) -- (y) node [above,midway] {$A_1$};
\draw[->] (u) .. controls (1.5,-0.9) and (0.5,-0.9) .. node [auto,pos=0.15] {$T^t$} (u);

\node at (3,0) (x1) {$\mathbf{x}$};
\node at (4,0) (u1) {$\left[\!\begin{smallmatrix}d_2\\ 0\end{smallmatrix}\!\right]$};
\node at (5,0) (y1) {$\mathbf{y}$};

\draw[->] (x1) -- (u1) node [above,midway] {$B_2$};
\draw[->] (u1) -- (y1) node [above,midway] {$A_2$};
\draw[->] (u1) .. controls (4.5,-0.9) and (3.5,-0.9) .. node [auto,pos=0.15] {$T^t$} (u1);
\end{tikzpicture}
\end{center}
\end{lemma}

\begin{proof}
See Section \ref{sec:ind} of the Appendix.
\end{proof}

By Lemma \ref{lem:ind}, we can choose any matrix $C$ from $\SL$ that maps a vector $\mathbf{x}$ to the vector $\begin{bmatrix}1\\ 0\end{bmatrix}$, and for each $\mathbf{y}$ on the line $L$ we can choose any matrix $B_\mathbf{y}$ that maps $\begin{bmatrix}1\\ 0\end{bmatrix}$ to the vector $\mathbf{y}$. Then the solution of Equation (\ref{eq:sc}) will be described by the following set
$
\F=\{B_\mathbf{y}T^tC\ :\ \mathbf{y}\in L \text{ and } t\in \Z\}.
$
Figure \ref{fig:scalar} gives geometric interpretation of this solution.
\begin{figure}[h]
\centering
\begin{tikzpicture}[shorten >=1pt,scale=1.3]
\tikzset{dots/.style={circle,draw,fill,inner sep=0pt,minimum size=4pt}}

\draw[->,semithick] (0,-1) -- (0,2);
\draw[->,semithick] (-1.5,0) -- (4.5,0);
\draw[thick,blue] (-0.8,2) -- (1,-1);

\node (x) at (4,1) [dots,label=right:$\mathbf{x}$] {};
\node (y) at (-0.5,1.5) [dots,label=left:$\mathbf{y}$] {};
\node (z) at (1.5,0) [dots,label=below:${(1,0)}$] {};

\node at (0.5,-0.5) {$L$};

\draw[->,semithick] (x) to node [below,midway] {$C$} (z);
\draw[->,semithick] (z) to node [below,midway] {$B_\mathbf{y}$} (y);
\draw[->,semithick] (z) .. controls (0.5,1.5) and (2.5,1.5).. node [auto,pos=0.7] {$T^t$} (z);

\end{tikzpicture}
\caption{Geometric interpretation of the scalar reachability problem.}
\label{fig:scalar}
\end{figure}
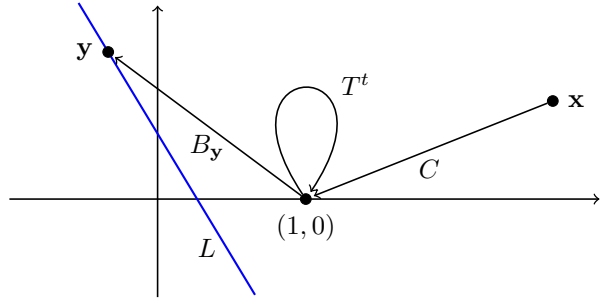

We need to choose $B_\mathbf{y}$ in such a way that $\F$ becomes a regular set. Let $\mathbf{y}=\begin{bmatrix}y_1\\ y_2\end{bmatrix}\in L$, then we have $ay_1+y_2=1$. As one can check, if we let $B_\mathbf{y}=\begin{bmatrix}y_1 & -1\\ -ay_1+1 & a\end{bmatrix}$ then $B_\mathbf{y}\in \SL$ and $B_\mathbf{y}\begin{bmatrix}1\\ 0\end{bmatrix} = \mathbf{y}$.

Since every entry of $B_\mathbf{y}$ is a linear function of $y_1$, we obtain by Proposition \ref{regmat} that $B_\mathbf{y} = AT^{ky_1}D$, where $A$ and $D$ are some matrices from $\SL$ and $k$ is some integer number (in fact, one can show that $k=1$). Finally, we can write all solutions of Equation~(\ref{eq:sc}) as
$
\F=\{AT^{ky_1}DT^tC\ :\ y_1\in \Z \text{ and } t\in \Z\}
$.
This is clearly a regular set and, therefore, the scalar reachability problem is decidable.
\end{proof}

\bibliography{refs}

\appendix
\section{Proof of Proposition \ref{linmat}}\label{sec:linmat}

For the proof of Proposition \ref{linmat} we will need the following two lemmas.

\begin{lemma} \label{lineq}
Consider a linear congruence equation $ax\equiv b \pmod n$. If $\gcd(a,n) \nmid b$, then the equation has no solution. If $\gcd(a,n) \mid b$, then all solutions of the equation can be written in the form $x\equiv c \pmod {\frac{n}{\gcd(a,n)}}$ for some $c$. Moreover, there is a polynomial time algorithm that determines whether such equation has a solution and if so, finds it.
\end{lemma}

\begin{proof}
Given $a$ and $n$, using Euclidean algorithm we can find in polynomial time $d=\gcd(a,n)$ and integer numbers $u$ and $v$ such that $d=ua+vn$. Equation $ax\equiv b \pmod n$ can be written as $ax=b+kn$, where $k\in \Z$. It is clear that if $d\nmid b$, then there is no solution. Otherwise, let $b=b'd$, $a=a'd$, and $n=n'd$. Then our equation is equivalent to $a'x\equiv b' \pmod {n'}$. Furthermore, we have $ua'+vn'=1$ and hence $ua'\equiv 1\pmod {n'}$. Thus
\[
x\equiv (ua')x\equiv u(a'x)\equiv ub'\pmod {n'}.
\]
Note that all these computations can be done in PTIME.
\end{proof}

\begin{lemma} \label{linsys}
Consider a system of two linear congruence equations
\begin{align}
\label{sys}
a_1x \equiv b_1 \pmod {n_1}&&
a_2x \equiv b_2 \pmod {n_2}
\end{align}
Such system either has no solution, or all its solutions are of the form $x\equiv c\pmod n$ for some $c$ and $n\mid n_1n_2$. Moreover, there is a polynomial time algorithm that determines whether (\ref{sys}) has a solution and if so, finds it.
\end{lemma}

\begin{proof}
Using the algorithm of Lemma \ref{lineq}, we can solve each equation separately. If one of them does not have a solution, then the system (\ref{sys}) also has no solution. Suppose the first and second equation have the solutions $x\equiv c_1 \pmod {n_1'}$ and $x\equiv c_2 \pmod {n_2'}$, respectively, which can be found in PTIME. Note that $n_i'\mid n_i$ for $i=1,2$.

Let $n=\lcm(n_1',n_2')$. We can rewrite the solutions as
\[
\begin{split}
x &\equiv c_1,\ c_1+n_1',\ c_1+2n_1',\ \ldots,\ c_1+(n_2''-1)n_1' \pmod n,\\
x &\equiv c_2,\ c_2+n_2',\ c_2+2n_2',\ \ldots,\ c_2+(n_1''-1)n_2' \pmod n,
\end{split}
\]
where $n_1''=n/n_2'$ and $n_2''=n/n_1'$. Let
\[
\begin{split}
A_1 &= \{c_1,\ c_1+n_1',\ c_1+2n_1',\ \ldots,\ c_1+(n_2''-1)n_1'\},\\
A_2 &= \{c_2,\ c_2+n_2',\ c_2+2n_2',\ \ldots,\ c_2+(n_1''-1)n_2'\}.
\end{split}
\]
Note that $A_1\cap A_2$ contains at most one element. Indeed, if $c,c'\in A_1\cap A_2$, then $n_1'\mid c-c'$ and $n_2'\mid c-c'$. Hence $n=\lcm(n_1',n_2')\mid c-c'$. Since $|c-c'|<n$, we have $c=c'$.

Now if $A_1\cap A_2$ is empty, then (\ref{sys}) has no solution. If $A_1\cap A_2 = \{c\}$, then the solution of (\ref{sys}) is $x\equiv c\pmod n$. To find this solution in PTIME, observe the following. The equations $x\equiv c_1 \pmod {n_1'}$ and $x\equiv c_2 \pmod {n_2'}$ are equivalent to $x=c_1+kn_1'$ and $x=c_2+ln_2'$, respectively, where $k,l\in \Z$. To find the intersection of these solutions we set $c_1+kn_1'=c_2+ln_2'$, which is equivalent to $c_1-c_2=ln_2'-kn_1'$. Using Euclidean algorithm, we can find in PTIME $d=\gcd(n_1',n_2')$ and integer numbers $u,v$ such that
\begin{equation} \label{eq:gcd}
d=un_1'+vn_2'.
\end{equation}
Obviously, if $d\nmid c_1-c_2$, then there is no solution. So suppose $c_1-c_2=hd$, for some $h\in \Z$. Multiplying (\ref{eq:gcd}) by $h$ we obtain
\[
c_1-c_2 = hd = (hu)n_1'+(hv)n_2'\quad \text{or}\quad
c_1-(hu)n_1' = c_2+(hv)n_2'.
\]
Let $c$ be the number in the set $\{0,\ldots,n-1\}$ such that
\[
c\equiv c_1-(hu)n_1' = c_2+(hv)n_2' \pmod n.
\]
Then $x\equiv c\pmod n$ is the desired solution. It is not hard to see that the above algorithm runs in polynomial time.
\end{proof}

\setcounter{thmold}{\value{theorem}}
\setcounter{theorem}{\value{thm1}}
\begin{proposition}
  Let $\mathbf{x}=\begin{bmatrix}x_1\\ x_2\end{bmatrix}$ and $\mathbf{y}=\begin{bmatrix}y_1\\ y_2\end{bmatrix}$ be two vectors from $\Z\times\Z$, such that $\mathbf{x}$ is not equal to the zero vector\/ $\mathbf{0}$, and consider the matrix equation $M\mathbf{x}=\mathbf{y}$, where $M$ is an unknown matrix from $\SL$. Then either this equation does not have a solution or all its solutions are given by $M=tA_1+A_2$, where $t$ is any integer number, $A_1,A_2$ are some matrices from $\Z^{2\times 2}$ such that $A_1$ is a nonzero matrix. Moreover, there is a polynomial time algorithm that determines whether such equation has a solution and if so, finds it.
\end{proposition}
\setcounter{theorem}{\value{thmold}}

\begin{proof}
Let $M=\begin{bmatrix} a & b\\ c & d\end{bmatrix}$ be a matrix that satisfies the equations $M\mathbf{x}=\mathbf{y}$ and $\det(M)=1$. So we have the following system of equations:
\begin{align}\label{eq1}
 x_1a+x_2b = y_1 &&
 x_1c+x_2d = y_2 &&
 ad-bc = 1
\end{align}

Recall that by assumption $\mathbf{x}\neq \mathbf{0}$. Without loss of generality, suppose that $x_1\neq 0$. In this case we have
$a = \dfrac{y_1-x_2b}{x_1}$, $c = \dfrac{y_2-x_2d}{x_1}$.
Substituting these values for $a$ and $c$ into the equation $ad-bc = 1$, we obtain
$
  (y_1-x_2b)d-(y_2-x_2d)b=x_1
$
or, equivalently, $y_1d-y_2b = x_1$. If $y_1=y_2=0$, then there is no solution because by assumption $x_1\neq 0$. Again, without loss of generality, assume that $y_1\neq 0$. Hence $d = \dfrac{x_1+y_2b}{y_1}$. If we choose $b$ as a free parameter, then the general solution of the system of equations (\ref{eq1}) will be:
\[
   a = \frac{y_1-x_2b}{x_1},\quad  d = \frac{x_1+y_2b}{y_1},\quad
   c = \frac{y_2-x_2\frac{x_1+y_2b}{y_1}}{x_1} = \frac{y_1y_2-x_1x_2-x_2y_2b}{x_1y_1}.
\]
We are interested only in integer solutions, that is, when $a$, $c$, and $b$ are in $\Z$, which means that $b$ must satisfy the following congruences:
\begin{align*}
  x_2b \equiv y_1\ (\bmod\ {x_1})&&
  y_2b \equiv -x_1\ (\bmod\ {y_1})&&
  x_2y_2b \equiv y_1y_2-x_1x_2\ (\bmod\ {x_1y_1})
\end{align*}
Applying the algorithm from Lemma \ref{linsys} two times, we can determine in PTIME whether the above system has a solution or not. If the solution exists, the algorithm outputs it in the form $b\equiv b_2 \pmod {b_1}$, where $b_1\mid x_1y_1$.

So, the coefficient $b$ is of the form $b = b_1t+b_2$, where $t\in \Z$. Substituting this expression for $b$ in the formulas for $a$, $c$, and $d$ we obtain:
\[
  \begin{split}
   a &= \frac{y_1-x_2b_2-x_2b_1t}{x_1} = a_1t+a_2,\qquad
   d = \frac{x_1+y_2b_2+y_2b_1t}{y_1} = d_1t+d_2,\\
   c &= \frac{y_1y_2-x_1x_2-x_2y_2b_2-x_2y_2b_1t}{x_1y_1} = c_1t+c_2,
  \end{split}
\]
where $a_i$, $c_i$, and $d_i$, for $i=1,2$, are some constants which are necessarily in $\Z$ because if we let $t=0$ or $t=1$ in the above expressions they must evaluate to integer numbers. Therefore, the solution to the system of equations (\ref{eq1}) can be written as:\\
\null\qquad
$
  M = \begin{bmatrix} a_1t+a_2 & b_1t+b_2\\ c_1t+c_2 & d_1t+d_2 \end{bmatrix} = 
  t\begin{bmatrix} a_1 & b_1\\ c_1 & d_1 \end{bmatrix} + 
   \begin{bmatrix} a_2 & b_2\\ c_2 & d_2 \end{bmatrix}\!,
$\\
where $t$ is any integer number. To complete the proof we set $A_1=\begin{bmatrix} a_1 & b_1\\ c_1 & d_1 \end{bmatrix}$ and $A_2=\begin{bmatrix} a_2 & b_2\\ c_2 & d_2 \end{bmatrix}$. Note that $A_1$ is a nonzero matrix since at least one of its coefficients, namely $b_1$, is not equal to zero. Furthermore, the above algorithm runs in polynomial time because the only nontrivial step is to solve the system of linear congruence equations, which according to Lemma \ref{linsys} can be done in PTIME.
\end{proof}

\section{Proof of Proposition \ref{autom}}\label{sec:autom}

\setcounter{thmold}{\value{theorem}}
\setcounter{theorem}{\value{thm2}}
\begin{proposition}
There is an algorithm that for any given regular signed languages $L_1$ and~$L_2$ over the alphabet $\Sigma$, decides whether $\phi(L_1)\cap\phi(L_2)$ is empty or not. 
\end{proposition}
\setcounter{theorem}{\value{thmold}}

\begin{proof}

We will describe a construction that turns any signed automaton $\A$ over $\Sigma$ into a new signed automaton $\wt{\A}$ such that
\begin{itemize}
\item $\phi(L(\wt{\A}))=\phi(L(\A))$ and

\item for every $M\in \phi(L(\wt{\A}))$, there is a reduced word $w$ such that $M=\phi(w)$ or $M=-\phi(w)$ and $w\in L(\wt{\A})^+$ or $w\in L(\wt{\A})^-$, respectively.
\end{itemize}

Suppose $\A=(\Sigma,Q,I,\Delta,F^+,F^-)$, then $\wt{\A}$ is defined as follows $\wt{\A}=(\Sigma,\wt{Q},\wt{I},\wt{\Delta},\wt{F}^+,\wt{F}^-)$, where
\begin{itemize}
\item $\wt{Q}=Q\times\{+,-\}$,
\item $\wt{I}=I\times\{+\}$,
\item $\wt{F}^+=\{(q,+) :\, q\in F^+\}\cup \{(q,-) :\, q\in F^-\}$,
\item $\wt{F}^-=\{(q,+) :\, q\in F^-\}\cup \{(q,-) :\, q\in F^+\}$,
\end{itemize}
To define $\wt{\Delta}$, we first set $\wt{\Delta} = \Delta$. Then for each transition $(q_1,X,q_2)\in \Delta$, we add the following two transition into $\wt{\Delta}$: $((q_1,+),X,(q_2,+))$ and $((q_1,-),X,(q_2,-))$.

Furthermore, we iteratively add new $\varepsilon$-transitions to $\wt{\Delta}$ as follows: if there is a run of $\wt{\A}$ from $(q_1,s_1)$ to $(q_2,s_2)$ labelled by $SS$ or $RRR$, then we add an $\varepsilon$-transition from $(q_1,s_1)$ to $(q_2,\bar{s}_2)$, where $\bar{s}_2$ is the sign opposite to $s_2$. For instance, if there is a run from $(q_1,+)$ to $(q_2,+)$ labelled by $RRR$, then we add an $\varepsilon$-transition from $(q_1,+)$ to $(q_2,-)$ (see Figure~\ref{fig:aut} for an illustration). We continue this process until no new $\varepsilon$-transitions can be added.
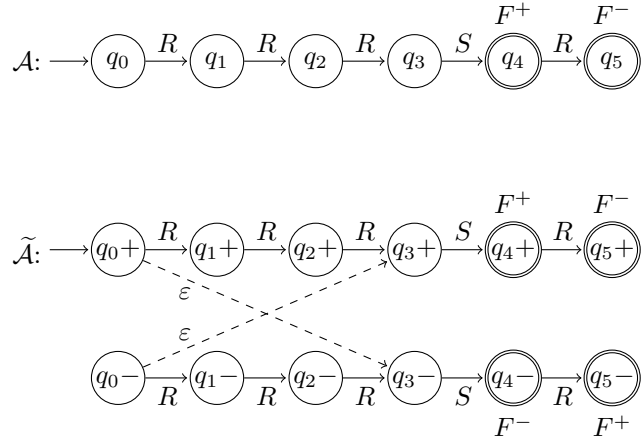
\begin{figure}[h]
\centering
\begin{tikzpicture}[shorten >=1pt,node distance=1.7cm and 1.3cm,on grid,auto] 
\tikzset{every state/.style={inner sep=0pt,minimum size=20pt}}
   \node at (-2.5,0) {$\A$:};
   \node at (-2.5,-2.5) {$\wt{\A}$:};
   \node[state] (q_0) at (-1.3,0)  {$q_0$}; 
   \node[state,] (q_1) [right=of q_0] {$q_1$}; 
   \node[state] (q_2) [right=of q_1] {$q_2$}; 
   \node[state] (q_3) [right=of q_2] {$q_3$}; 
   \node[state,accepting,label={$F^+$}] (q_4) [right=of q_3] {$q_4$}; 
   \node[state,accepting,label={$F^-$}] (q_5) [right=of q_4] {$q_5$}; 
    \path[->] 
    (-2.2,0) edge  (q_0)
    (q_0) edge  node {$R$} (q_1)
    (q_1) edge  node {$R$} (q_2)
    (q_2) edge  node {$R$} (q_3)
    (q_3) edge  node {$S$} (q_4)
    (q_4) edge  node {$R$} (q_5);

   \node[state] (q_0+) at (-1.3,-2.5) {$q_0+$}; 
   \node[state] (q_1+) [right=of q_0+] {$q_1+$}; 
   \node[state] (q_2+) [right=of q_1+] {$q_2+$}; 
   \node[state] (q_3+) [right=of q_2+] {$q_3+$}; 
   \node[state,accepting,label={$F^+$}] (q_4+) [right=of q_3+] {$q_4+$}; 
   \node[state,accepting,label={$F^-$}] (q_5+) [right=of q_4+] {$q_5+$}; 
   \node[state] (q_0-) [below=of q_0+]  {$q_0-$}; 
   \node[state] (q_1-) [right=of q_0-] {$q_1-$}; 
   \node[state] (q_2-) [right=of q_1-] {$q_2-$}; 
   \node[state] (q_3-) [right=of q_2-] {$q_3-$}; 
   \node[state,accepting,label={below:$F^-$}] (q_4-) [right=of q_3-] {$q_4-$}; 
   \node[state,accepting,label={below:$F^+$}] (q_5-) [right=of q_4-] {$q_5-$}; 
    \path[->] 
    (-2.2,-2.5) edge (q_0+)
    (q_0+) edge  node {$R$} (q_1+)
    (q_1+) edge  node {$R$} (q_2+)
    (q_2+) edge  node {$R$} (q_3+)
    (q_3+) edge  node {$S$} (q_4+)
    (q_4+) edge  node {$R$} (q_5+);
    \path[->] 
    (q_0-) edge  node[below] {$R$} (q_1-)
    (q_1-) edge  node[below] {$R$} (q_2-)
    (q_2-) edge  node[below] {$R$} (q_3-)
    (q_3-) edge  node[below] {$S$} (q_4-)
    (q_4-) edge  node[below] {$R$} (q_5-);
    \path[->, dashed] 
    (q_0+) edge  node[pos=0.17, below] {$\varepsilon$} (q_3-);
    \path[->, dashed] 
    (q_0-) edge  node[pos=0.17, above] {$\varepsilon$} (q_3+);
\end{tikzpicture}
\caption{An example of an automaton $\A$ (above) and its corresponding automaton $\wt{\A}$ (below). The final states from $F^+$ and $F^-$ are marked by the labels $F^+$ and $F^-$, respectively.}
\label{fig:aut}
\end{figure}

Note that in $\SL$ we have $\phi(S)^2=\phi(R)^3=-I$, and this is reflected in the change of sign of the end state of a new $\varepsilon$-transition. It is not hard to see that $\wt{\A}$ is indeed the desired automaton.

Let $\A_1$ and $\A_2$ be two finite signed automata such that $L(\A_1)=L_1$ and $L(\A_2)=L_2$. To check whether $\phi(L_1)\cap\phi(L_2)$ is empty or not, we take the automata $\A_1$ and $\A_2$ and construct the new automata $\wt{\A}_1$ and $\wt{\A}_2$ as described above.

Now we can prove that
$\phi(L_1)\cap\phi(L_2)\neq \emptyset$ if and only if $L(\wt{\A}_1)^+\cap L(\wt{\A}_2)^+\neq \emptyset$ or $L(\wt{\A}_1)^-\cap L(\wt{\A}_2)^-\neq \emptyset$.
Indeed, suppose that $M\in \phi(L_1)\cap\phi(L_2)$. By the above construction we have $\phi(L(\wt{\A}_i))= \phi(L_i)$, for $i=1,2$, and there is a reduced word $w\in \Sigma^*$ such that $M=\phi(w)$ or $M=-\phi(w)$ and $w\in L(\wt{\A}_i)^+$ or $w\in L(\wt{\A}_i)^-$, respectively, for both $i=1,2$. In the fist case we have $w\in L(\wt{\A}_1)^+\cap L(\wt{\A}_2)^+$ and in the second case $w\in L(\wt{\A}_1)^-\cap L(\wt{\A}_2)^-$. The implication in the other direction is trivial.

To complete the proof we note that the intersection of regular languages is again regular, and the emptiness problem for regular languages is decidable.
\end{proof}

\section{Proof of Lemma \ref{lem:ind}}\label{sec:ind}

\setcounter{thmold}{\value{theorem}}
\setcounter{theorem}{\value{thm3}}
\begin{lemma}
Let $\mathbf{x}=\begin{bmatrix}x_1\\ x_2\end{bmatrix}$ and $\mathbf{y}=\begin{bmatrix}y_1\\ y_2\end{bmatrix}$ be any vectors from $\Z\times \Z$ such that $\gcd(x_1,x_2)=\gcd(y_1,y_2)=d$. Let $d_1$ and $d_2$ be any integer numbers with $|d_1|=|d_2|=d$ and let $A_1$, $B_1$ and $A_2$, $B_2$ by any matrices from $\SL$ such that $B_i\mathbf{x}=\begin{bmatrix}d_i\\ 0\end{bmatrix}$ and $A_i\begin{bmatrix}d_i\\ 0\end{bmatrix}= \mathbf{y}$, for $i=1,2$. Then
$
\{A_1T^tB_1\ :\ t\in \Z\} = \{A_2T^tB_2\ :\ t\in \Z\}.
$
In other words, the following diagrams define the same set of matrices that map $\mathbf{x}$ to $\mathbf{y}$.
\begin{center}
\begin{tikzpicture}[scale=1.3]
\node at (0,0) (x) {$\mathbf{x}$};
\node at (1,0) (u) {$\left[\!\begin{smallmatrix}d_1\\ 0\end{smallmatrix}\!\right]$};
\node at (2,0) (y) {$\mathbf{y}$};

\draw[->] (x) -- (u) node [above,midway] {$B_1$};
\draw[->] (u) -- (y) node [above,midway] {$A_1$};
\draw[->] (u) .. controls (1.5,-0.9) and (0.5,-0.9) .. node [auto,pos=0.15] {$T^t$} (u);

\node at (3,0) (x1) {$\mathbf{x}$};
\node at (4,0) (u1) {$\left[\!\begin{smallmatrix}d_2\\ 0\end{smallmatrix}\!\right]$};
\node at (5,0) (y1) {$\mathbf{y}$};

\draw[->] (x1) -- (u1) node [above,midway] {$B_2$};
\draw[->] (u1) -- (y1) node [above,midway] {$A_2$};
\draw[->] (u1) .. controls (4.5,-0.9) and (3.5,-0.9) .. node [auto,pos=0.15] {$T^t$} (u1);
\end{tikzpicture}
\end{center}
\end{lemma}
\setcounter{theorem}{\value{thmold}}

\begin{proof}
Let us define $\F_i=\{A_iT^tB_i\ :\ t\in \Z\}$ for $i=1,2$. We need to show that $\F_1=\F_2$. Suppose that $M\in \F_1$, that is, $M=A_1T^{t_1}B_1$ for some $t_1\in \Z$. We want to show that $M\in \F_2$, that is, $M=A_2T^{t_2}B_2$ for some $t_2\in \Z$. First, let us write $M$ as\\
\null\qquad
$
M=A_2(A_2^{-1}A_1)T^{t_1}(B_1B_2^{-1})B_2.
$\\
Note that
$
A_2^{-1}A_1\begin{bmatrix}d_1\\ 0\end{bmatrix} = \begin{bmatrix}d_2\\ 0\end{bmatrix}$ and
$B_1B_2^{-1}\begin{bmatrix}d_2\\ 0\end{bmatrix} = \begin{bmatrix}d_1\\ 0\end{bmatrix}$.
Suppose
$
A_2^{-1}A_1=\begin{bmatrix}a_1 & a_2\\ a_3 & a_4\end{bmatrix}$ and
$B_1B_2^{-1}=\begin{bmatrix}b_1 & b_2\\ b_3 & b_4\end{bmatrix}$.
Then we have $a_1d_1=d_2$, $a_3d_1=0$ and $b_1d_2=d_1$, $b_3d_2=0$. Hence $a_3=b_3=0$.

Now we need to consider two cases: (1) $d_1=d_2$ and (2) $d_1=-d_2$. In the first case $a_1=b_1=1$ and in the second case $a_1=b_1=-1$. Note that $A_2^{-1}A_1$ and $B_1B_2^{-1}$ are matrices from $\SL$. Therefore, in the first case we must have that\\
\null\qquad$
A_2^{-1}A_1=\begin{bmatrix}1 & a_2\\ 0 & 1\end{bmatrix} \text{ and }\
B_1B_2^{-1}=\begin{bmatrix}1 & b_2\\ 0 & 1\end{bmatrix}
$\\
and in the second case we must have that
$
A_2^{-1}A_1=\begin{bmatrix}-1 & a_2\\ 0 & -1\end{bmatrix} \text{ and }\
B_1B_2^{-1}=\begin{bmatrix}-1 & b_2\\ 0 & -1\end{bmatrix}.
$
So, in the first case we obtain\\
\null\qquad$
M=A_2T^{a_2}T^{t_1}T^{b_2}B_2=A_2T^{a_2+t_1+b_2}B_2,
$\\
and in the second case we have\\
\null\qquad$
M=A_2(-T^{-a_2})T^{t_1}(-T^{-b_2})B_2=A_2T^{-a_2+t_1-b_2}B_2.
$\\
Hence in both cases $M\in \F_2$. Similarly, we can show that if $M\in \F_2$, then $M\in \F_1$.
\end{proof}

\section{Proof of Corollary \ref{regmatcor}} \label{sec:cor}

\setcounter{thmold}{\value{theorem}}
\setcounter{theorem}{\value{thm4}}
\begin{corollary}
The value of the parameter $k$ in Theorem~\ref{regmatthm} is equal to 1.
\end{corollary}
\setcounter{theorem}{\value{thmold}}

\begin{proof}
Our proof will rely on Proposition \ref{prop:uv} from Section \ref{sec:geom}. We need to show that
\[
\{BT^{kt}C : t\in \Z\} = \{BT^{t}C : t\in \Z\}.
\]
Note that the inclusion $\{BT^{kt}C : t\in \Z\} \subseteq \{BT^{t}C : t\in \Z\}$ is obvious. On the other hand, by Proposition \ref{prop:uv}, the vectors $\mathbf{u}=C\mathbf{x}$ and $\mathbf{v}=B^{-1}\mathbf{y}$ are equal to each other and have the form $\begin{bmatrix}d\\ 0\end{bmatrix}$. Since $T=\begin{bmatrix} 1 & 1\\ 0 & 1 \end{bmatrix}$, it is easy to check that $T\mathbf{u}=\mathbf{u}=\mathbf{v}$. From this we can conclude that $T^t\mathbf{u}=\mathbf{v}$ and hence $BT^{t}C\mathbf{x}=\mathbf{y}$ for any $t\in \Z$. In other words, matrices of the form $BT^{t}C$ transform $\mathbf{x}$ to $\mathbf{y}$.
However, in Theorem \ref{regmatthm} we proved that all such matrices belong to the set $\{BT^{kt}C\ :\ t\in \Z\}$. Therefore, we obtain the inclusion $\{BT^{t}C\ :\ t\in \Z\} \subseteq \{BT^{kt}C\ :\ t\in \Z\}$, and this proves the corollary.
\end{proof}

%%
%% Bibliography
%%

%% Either use bibtex (recommended), 

%% .. or use the thebibliography environment explicitely

\end{document}